\documentclass[a4paper,UKenglish,runningheads,11pt]{llncs}

\usepackage{tabularx,booktabs,multirow,delarray,array}
\usepackage{graphicx,amssymb,amsmath}
\usepackage[ruled,vlined,linesnumbered]{algorithm2e}
\usepackage{fullpage}
\usepackage{latexsym}
\usepackage{enumerate}
\usepackage{lineno}
\usepackage{subfigure}
\usepackage{wrapfig}
\usepackage{algorithmic}
\usepackage{hyperref}

\newenvironment{proof}{\par\noindent{\bf Proof:}}{\mbox{}\hfill$\qed$\\}

\newcommand{\ignore}[1]{ }

\newcounter{rem}
\def\etal{\textsl{et~al.}}

\def\qed{\hbox{\rlap{$\sqcap$}$\sqcup$}}

\begin{document}

\title{Constant Workspace Algorithms for Computing Relative Hulls in the Plane}
\titlerunning{Computing relative hulls using $O(1)$ workspace}

\author{
Himanshu Chhabra\inst{1}
\and
R. Inkulu\inst{1}
}

\institute{
Department of Computer Science and Engineering\\
Indian Institute of Technology Guwahati\\
\email{\{c.himanshu,rinkulu\}@iitg.ac.in}
}

\maketitle

\pagenumbering{arabic}
\setcounter{page}{1}

\begin{abstract}
Constant workspace algorithms use a constant number of words in addition to the read-only input to the algorithm.
In this paper, we devise algorithms to efficiently compute relative hulls in the plane using a constant workspace.
Specifically, we devise algorithms for the following three problems: \newline
(i) Given two simple polygons $P$ and $Q$ with $P \subset Q$, compute a simple polygon $P'$ with a perimeter of minimum length such that $P \subseteq P' \subseteq Q$. \newline
(ii) Given two simple polygons $P$ and $Q$ such that $Q$ does not intersect the relative interior of $P$ but it does intersect the relative interior of the convex hull of $P$, compute a weakly simple polygon $P'$ with a perimeter of minimum length such that $P \subseteq P'$, the convex hull of $P$ contains $P'$, and $P'$ does not intersect the relative interior of $Q$. \newline
(iii) Given a set $S$ of points located in a simple polygon $P$, compute a weakly simple polygon $P'$ with a perimeter of minimum length such that $P' \subseteq P$ and $P'$ contains all the points in $S$. \newline
To our knowledge, no prior work devised algorithms to compute relative hulls using a constant workspace, and this work is the first such attempt.
\end{abstract}

\section{Introduction}
\label{sect:intro}

A {\it simple polygon} in $\mathbb{R}^2$ is the closed region bounded by a piecewise linear simple closed curve.
A polygon $P$ with at least three sides is called a {\it weakly simple polygon} whenever each of the vertices of $P$ can be perturbed by at most $\epsilon$ to obtain a simple polygon for every real number $\epsilon > 0$.
A simple polygon $P$ in $\mathbb{R}^2$ is {\it convex} whenever the line segment $pq$ is contained in $P$ for every two points $p, q \in P$. 
Computing the convex hull in the plane is a fundamental problem in computational geometry. 
Specifically, the following two problems are of interest: 
(i) Given a set $S$ of $n$ points in the plane, find a convex (simple) polygon $CH(S)$ with perimeter of minimum length so that $CH(S)$ contains all the points in $S$.
(ii) Given a simple polygon $P$ in the plane defined by $n$ vertices, find a convex (simple) polygon $CH(P)$ with perimeter of minimum length that contains $P$.  
It is well known that both $CH(S)$ and $CH(P)$ are unique.
The famous algorithms for the first problem include Jarvis march (a specialization of Chand and Kapur's gift wrapping), Graham's scan, Quickhull algorithm, Shamos's merge algorithm, Preparata-Hong's merge algorithm, incremental and randomized incremental algorithms.
Many of these algorithms are detailed in popular textbooks on computational geometry, see Preparata and Shamos~\cite{books/compgeom/prep1985} and de Berg~\etal~\cite{books/compgeom/deberg2008}.
For $h$ being the number of points in $S$ that are on the boundary of $CH(S)$, the worst-case lower bound on the time complexity of the first problem is known to be $\Omega(n\lg{h})$.
For this problem, several optimal algorithms that take $\Theta(n\lg{h})$ worst-case time are known: for example, algorithms by Kirkpatrick and Seidel~\cite{journals/siamjc/1986KirkSeid} and Chan~\cite{journals/dcg/1996Chan}.
For the second problem, several algorithms with $O(n)$ time complexity are known; the one by Lee is presented in \cite{books/compgeom/prep1985}.

The relative hull, also known as the geodesic hull, has received increasing attention in computational geometry, and it appears in a variety of applications in robotics, industrial manufacturing, and geographic information systems.
The relative hulls and their related structures based on geodesic metrics have been used to approximate curves and surfaces in digital geometry.
The following two specific problems on relative hulls are famous:
(i) Given two simple polygons $P$ and $Q$ with $P \subset Q$, compute a simple polygon, known as the {\it relative hull of $P$ with respect to $Q$}, denoted by $RH(P|Q)$, with a perimeter of minimum length such that $P \subseteq RH(P|Q) \subseteq Q$.
(ii) Given a set $S$ of points located in a simple polygon $Q$, compute a weakly simple polygon, known as the {\it relative hull of $S$ with respect to $Q$}, denoted by $RH(S|Q)$, with a perimeter of minimum length such that $RH(S|Q)$ contains all the points in $S$ and is contained in $Q$.
Toussaint~\cite{journals/sigproc/1986toussaint} devised algorithms for both of these problems.

In this paper, we consider three problems on relative hulls and devise algorithms to compute them using a constant workspace.
The workspace of an algorithm is the additional space, excluding both the input and output space complexities, required to execute that algorithm.
Traditionally, workspace complexity has played only second fiddle to time complexity. 
However, due to the limited space on chips, algorithms need to use space efficiently, for example, in embedded systems.
For such reasons, algorithms using a small footprint are desired.
A detailed survey of various models and famous algorithms in these models can be found in Banyassady~\etal~\cite{journals/sigact/2018BanyKM}.
Next, we mention three computational models that are popular in designing algorithms focusing on workspace efficiency: algorithms that work using a constant workspace, in-place algorithms, and algorithms that have a trade off between time and workspace.

In constant workspace algorithms, input is read-only, output is not stored but streamed, and these algorithms use $O(1)$ workspace.
The following are the well-known constant workspace algorithms:
computing shortest paths in trees and simple polygons by Asano~\etal~\cite{journals/jgaa/2011AsanoMW},
triangulating a planar point set and trapezoidal decomposition of a simple polygon by Asano~\etal~\cite{journals/jocg/2011AsanoMRW}, and finding common tangents of two polygons that separate those polygons by Abrahamsen~\cite{conf/socg/2015Abrahamsen}.
Reingold~\cite{journals/jacm/2008reingold} devised a constant workspace algorithm to determine whether a path exists between any two given nodes in an undirected graph, settling an important open problem.
As mentioned, in this paper, we devise algorithms for computing relative hulls in the plane, and each of these algorithms uses constant workspace.

In-place algorithms also use $O(1)$ workspace; however, the input array is not read-only.
That is, at any instant during the algorithm's execution, the input array can contain any permutation of elements of the initial input array.
Strictly speaking, these algorithms use $O(n)$ workspace, where $n$ is the input size.
Some of the notable in-place algorithms from the literature include,
computing convex hulls in the plane by Bronnimann~\etal~\cite{conf/lati/2002BronniIKMMT}, 
finding points of intersections of line segments in the plane by Chen and Chan~\cite{conf/cccg/2003ChenChan},
computing convex hulls in $\mathbb{R}^3$ and Delaunay triangulations in the plane by Bronnimann~\etal~\cite{conf/scg/2004BronniCC},
sorting algorithm by Franceschini and Geffert~\cite{journals/jacm/2005FranGeff},
nearest neighbour search in the plane by Chan and Chen~\cite{conf/soda/2008ChanChen},
computing 3-d convex hulls and computing line segment intersection in the plane by Chan and Chen~\cite{journals/cgta/2010ChanChen},
and
computing layers of maxima in the plane and the maxima in $\mathbb{R}^3$ by Blunck and Vahrenhold~\cite{journals/algorithmica/2010BlunckVahr}.

For algorithms that provide a trade-off between time and workspace, $s$ is given as an additional input parameter, and the algorithms are permitted to use $O(s)$ workspace in the worst case.
The time complexity of such algorithms is expressed as a function of $n$ and $s$, making it viable to provide an algorithmic scheme to achieve a trade-off between workspace and time complexities, a lower asymptotic worst-case time complexity as $s$ grows and a larger asymptotic worst-case time complexity as $s$ gets smaller.
The notable algorithms in this model include, 
computing all-nearest-larger-neighbors by Asano~\etal~\cite{journals/cgta/2013AsanoBBKMRS},
computing the planar convex hull by Darwish and Elmasry~\cite{conf/esa/2014DarwElma},
computing shortest paths in a simple polygon by Asano~\etal~\cite{journals/cgta/2013AsanoBBKMRS} and by Peled~\cite{conf/socg/2015Peled},
triangulating a simple polygon by Aronov~\etal~\cite{journals/jocg/2017AronovKPRR},
triangulating points located in the plane by Korman~\etal~\cite{journals/cgta/2015KormanMRRSS},
computing the Voronoi diagram of points in the plane by Banyassady~\etal~\cite{journals/jocg/2018BanyKMRRSS},
computing the $k$-visibility region of a point in a simple polygon by Bahoo~\etal~\cite{journals/tcs/2019BahooBBDM}, and
computing the Euclidean minimum spanning tree of points in the plane by Banyassady~\etal~\cite{journals/jocg/2020BanyBM}.
Barba~\etal~\cite{journals/algorithmica/2015BarbaKLSS} devised a framework that trades off between time and workspace for applications that use the stack data structure.
Note that when $s$ is $1$, these algorithms are essentially constant workspace or in-place algorithms, depending on whether the input array is read-only or the input elements can be permuted.
A variation of this model considers $O(s\lg{n})$ number of bits as the workspace.

\subsection*{Preliminaries}
\label{subsect:prelim}

The convex hull of a set $S$ of points in the plane is denoted by $CH(S)$.
The convex hull of a simple polygon $P$, denoted by $CH(P)$, is the convex hull of the vertices of $P$.
The number of vertices of any (weakly) simple polygon $P$ is denoted by $|P|$.
For any (weakly) simple polygon $P$, the boundary (cycle) of $P$ is denoted by $bd(P)$.
The relative interior $rint(P)$ of a (weakly) simple polygon $P$ is $P$ excluding $bd(P)$.
We note that the relative interior of a weakly simple polygon consists of a collection of simply connected open regions.
For any edge $e$ of $CH(P)$ with $e \notin P$, the {\it pocket of $e$} is the simple polygon in $CH(P) \backslash P$ with $e$ on its boundary. 
And, $e$ is said to be the {\it lid of the pocket} of $e$. 
We call two points $p, q$ located in a (weakly) simple polygon $P$ {\it visible} to each other whenever the relative interior of the line segment joining $p$ and $q$ does not intersect any edge on $bd(P)$.
The vertices of an edge of a simple polygon are considered visible to each other.
Let $r'$ and $r''$ be two non-parallel rays with origin at a point $p$.
Let $\overrightarrow{v_1}$ and $\overrightarrow{v_2}$ be the unit vectors along rays $r'$ and $r''$, respectively. 
A {\it cone} $C_p(r', r'')$ is the set of points defined by rays $r'$ and $r''$ such that a point $q \in C_p(r', r'')$ if and only if $q$ can be expressed as a convex combination of vectors $\overrightarrow{v_1}$ and $\overrightarrow{v_2}$ with positive coefficients.
When the rays are evident from the context, we denote $C_p(r', r'')$ by $C_p$ or $C$.
Let $P$ be a simple polygon and let $T(P)$ be a triangulation of $P$.
Also, let $S$ be an ordered set of triangles in $T(P)$, with $|S| \ge 2$, such that every two successive triangles in $S$ are abutting along an edge of $T(P)$.
Then, the {\it sleeve} corresponding to $S$ is the simply connected closed region resulting from the union of all the triangles in $S$.
To distinguish from the vertices of a (weakly) simple polygon, every point in $\mathbb{R}^2$ that is not necessarily a vertex of a polygon is called a point, and the vertices of graphs are called nodes.

\subsection*{Our contributions}
\label{subsect:contrib}

Traditional algorithms for finding relative hulls proceed by designing data structures that use a workspace whose size is linear in the input complexity.
Since there is no restriction on the memory, where required, these algorithms save the intermediate structures, such as hulls or sub-polygons, into data structures. 
Saving information in these data structures helps to avoid recomputations.
However, the workspace constraints in a constant workspace setting do not permit the liberal use of such data structures to store processed information. 
Hence, the problem of computing relative hulls becomes harder when using only a constant workspace is allowed, and this needs engineering a solution with special techniques devoted to constant workspace algorithms.
This paper proposes algorithms for finding relative hulls in the plane for three problems.
Each of these algorithms uses a constant workspace.

The first problem asks to compute the relative hull of a simple polygon $P$, wherein $P$ is contained in another simple polygon $Q$.
The algorithm by Toussaint~\cite{journals/sigproc/1986toussaint} for this problem takes $O(|P|+|Q|)$ time while using $O(|P|+|Q|)$ workspace.
For this problem, Wiederhold and Reyesetal~\cite{journals/compimanal/2016WiederR} proposed an algorithm that takes $O(|P|^2+|Q|^2)$ time using $O(|P|+|Q|)$ workspace.
This result mainly resolves the incorrect cases in the recursive implementations given by Klette~\cite{conf/icivc/2010Klette}.
Our algorithm for this problem takes $O(|P|^2+|Q|^2)$ time in the worst case using $O(1)$ workspace.
The approach of our algorithm is distinct from algorithms devised for this problem.
Noting Jarvis march to compute the convex hull of a set of points uses $O(1)$ workspace, our algorithm invokes Jarvis march on the vertices of $P$ to compute the edges of $CH(P)$.
We modify the constant workspace algorithm given in Asano~\etal~\cite{journals/jocg/2011AsanoMRW} to compute the geodesic shortest path between two points in any simple polygon, and use this modified version to provide a constant workspace algorithm $\cal A$ for this problem.
For every edge $e$ of $CH(P)$ computed by Jarvis march, if $e$ does not intersect $rint(Q)$, our algorithm outputs $e$.
Otherwise, $\cal A$ is applied to compute the geodesic shortest path $\tau'$ between the vertices defining $e$ in $P' \backslash Q$, and our algorithm outputs each line segment of $\tau'$ as and when that edge is computed, instead of outputting $e$.
$RH(P|Q)$ is the closed region defined by all the line segments output by our algorithm.

Given two simple polygons $P$ and $Q$ such that $rint(P) \cap Q = \emptyset$ but $rint(CH(P)) \cap Q \ne \emptyset$, $Q$ is said to be located/placed {\it alongside} $P$.
The second problem intends to find the relative hull of a simple polygon $P$ when another simple polygon $Q$ is located alongside $P$.
Toussaint~\cite{journals/dcg/1989Toussaint} gave an algorithm for this problem that takes $O(|P|+|Q|)$ time and $O(|P|+|Q|)$ workspace.
When $Q$ is located alongside $P$, there are two cases that arise: $Q$ belongs to the relative interior of a pocket of $P$, or it does not.
The algorithm in \cite{journals/dcg/1989Toussaint} identifies and handles these two cases separately, and in both cases, the problem is reduced to finding the geodesic shortest path in a simple polygon.
Based on their algorithm, we propose an algorithm that takes $O(|P|^2+|Q|^2)$ time using $O(1)$ workspace for computing $RH(P | Q)$.
When $Q$ is not located in the relative interior of any pocket of $P$, following the algorithm for this problem in \cite{journals/dcg/1989Toussaint}, we reduce this problem to finding a geodesic shortest path in a simple polygon, and present a constant workspace algorithm to find the geodesic shortest path $\tau'$ in a simple polygon while that polygon is computed on-the-fly, as demanded by that shortest path finding algorithm.
This constant workspace algorithm primarily adapts the algorithm we present for the first problem.
When $Q$ is contained in a pocket $P'$ of $P$, we design a constant workspace algorithm to compute a geodesic shortest path $\tau'$ between the endpoints of the lid $e$ of $P'$ in $P' \backslash Q$.
This involves viewing $P' \backslash Q$ as a simple polygon $P''$, computing $P''$ on-the-fly from $P'$ and $Q$, and computing two geodesic shortest paths in $P''$ so that the union of these two paths is a geodesic shortest path $\tau'$ between the endpoints of $e$.
For computing these two geodesic shortest paths in $P''$, we use the constant workspace algorithm given in Asano~\etal~\cite{journals/jocg/2011AsanoMRW} with the modification presented in the previous algorithm.
In either of these two cases, the closed region bounded by $\tau'$ and the section of $bd(CH(P))$ between the endpoints of $\tau'$ such that this region contains $P$ is $RH(P|Q)$.
To achieve $O(1)$ workspace, like in the first algorithm, we compute the edges of the latter with Jarvis march, and output edges of $RH(P|Q)$ as they are computed.

The third problem aims to compute the relative hull of a set $S$ of points in a simple polygon $P$.
The algorithm by Toussaint~\cite{journals/sigproc/1986toussaint} for this problem takes $O((|P|+|S|)\lg{(|P|+|S|)})$ time but uses $O(|P|+|S|)$ workspace. 
Their algorithm computes a weakly simple polygon $P'$ by triangulating the simple polygon $P$, finding the convex hulls of points of $S$ lying in each triangle of that triangulation, and using the dual-tree of that triangulation to connect specific points on the boundaries of these hulls with geodesic shortest paths.
Then, a weakly simple polygon, located in $P \backslash P'$, having a perimeter of minimum length, is output as the $RH(S|P)$.
The constant workspace algorithm we present for this problem is more involved, though the outline of our algorithm is the same as the algorithm for this problem in \cite{journals/sigproc/1986toussaint}. 
As part of devising this algorithm, we engineered a solution by carefully combining constant workspace algorithms for several problems: Eulerian tour of trees, constrained Delaunay triangulation of a simple polygon, constructing the triangulation in an online fashion, computing the shortest geodesic path in a sleeve, and computing the sections of hulls of points lying in each triangle of a triangulation, to name a few.
Our algorithm takes $O(|P|^3+|S|^2)$ time in the worst case, and it uses $O(1)$ workspace.

In these three algorithms, at several junctures, we predominantly exploit the design paradigm that involves re-computing instead of storing the intermediate structures.
This paradigm is both natural and commonly used in space-constrained algorithms, wherein instead of storing information computed in a data structure, to save workspace, we recompute those elements as and when needed. 
To our knowledge, the algorithms presented in this paper are the first algorithms to compute relative hulls using constant workspace.

Section~\ref{sect:simpinsimp} presents an algorithm for computing the relative hull of a simple polygon located inside another simple polygon. 
Section~\ref{sect:simpalongsimp} gives an algorithm for computing the relative hull of a simple polygon $P$ when another simple polygon is placed alongside $P$.
An algorithm for computing the relative hull of a set of points located in a simple polygon is given in Section~\ref{sect:pointsinsimp}.
Conclusions are in Section~\ref{sect:conclu}.

\section{Relative hull of a simple polygon $P$ when $P$ contained in another simple polygon}
\label{sect:simpinsimp}

In this section, we devise a constant workspace algorithm to compute the relative hull $RH(P|Q)$ of a simple polygon $P$ contained in another simple polygon $Q$.
Here, $RH(P | Q)$ is a simple polygon such that its boundary cycle is contained in $Q \backslash rint(P)$ and it has a perimeter of minimum length.
It is well-known that $RH(P | Q)$ is unique for a given $P$ and $Q$.
As mentioned, to compute $RH(P | Q)$, the algorithm by Toussaint~\cite{journals/sigproc/1986toussaint} takes $O(|P|+|Q|)$ time using $O(|P|+|Q|)$ workspace, and the algorithm by Wiederhold and Reyesetal~\cite{journals/compimanal/2016WiederR} takes $O(|P|^2+|Q|^2)$ time using $O(|P|+|Q|)$ workspace.
The approach of our algorithm is different from these algorithms.
For $S'$ being the union of the set comprising the vertices of $P$ and the set comprising the vertices of $Q$, for convenience, we assume no two points in $S'$ have the same $x$-coordinates or $y$-coordinates, and no three points in $S'$ are collinear.

First, we note that every vertex of $CH(P)$ is a vertex of $RH(P|Q)$.
Hence, the leftmost vertex $p'$ of $P$ is a vertex of $RH(P|Q)$.
By traversing $bd(P)$ starting from $p'$, the edges of $bd(CH(P))$ are computed.
Since Jarvis march~\cite{journals/ipl/1973Jarvis} for computing the convex hull of a set of points in the plane uses $O(1)$ workspace, given any vertex $p_1$ on the hull, by applying an iteration of Jarvis march, our algorithm finds the vertex $p_2$ of $CH(P)$ that follows $p_1$ when $bd(CH(P))$ is traversed in clockwise direction.
Let $P'$ be the pocket of $P$ with $p_1p_2$ being the lid of $P'$.
For any edge $p_1p_2$ of $CH(P)$ computed by Jarvis march, if $p_1p_2$ has no intersection with $rint(Q)$, we output both $p_1$ and $p_2$, in that order.
Otherwise, we find the geodesic shortest path from $p_1$ to $p_2$ in simple polygon $P'' = P' \backslash rint(Q)$.

We modify the algorithm given by Asano~\etal~\cite{journals/jocg/2011AsanoMRW} to compute the geodesic shortest path from $p_1$ to $p_2$ in $P''$.
Noting the boundary of $P''$ is defined by sections of boundaries of $P$ and $Q$, the modification presented herewith involves computing sections of $P''$ in online fashion as and when required while maintaining structures (cones) required for the algorithm given in \cite{journals/jocg/2011AsanoMRW} to compute a shortest path from $p_1$ to $p_2$.
The algorithm for this problem in \cite{journals/jocg/2011AsanoMRW} maintains a vertex $p$ of $P''$ and two points $r'$ and $r''$ on $bd(P'')$ such that the line segments $pr'$ and $pr''$ lie in $P''$ with $pr''$ located counterclockwise from $pr'$. 
And, it proceeds by updating the triple $(p, r', r'')$ while ensuring at every juncture the geodesic shortest path from $p_1$ to $p_2$ passes through $p$ and the point $p_2$ lies in the subpolygon of $P''$ that is cut off by line segments $pr'$ and $pr''$.

\begin{figure}[ht]
\centering
\includegraphics[width=0.35\textwidth]{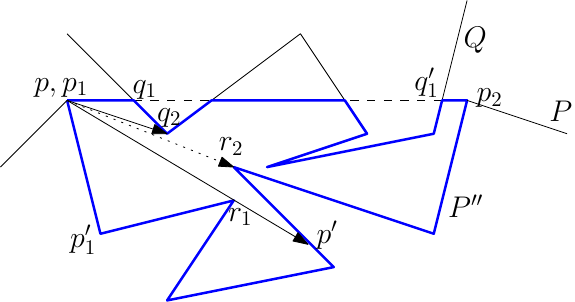}
\caption{Illustrating $P''$ in blue, whose boundary is made up of two polygonal chains:  
$p_1, q_1, q_2, \ldots, q_1', p_2$ consisting of sections of line segment $p_1p_2$ together with sections of $bd(Q)$, and 
$p_1, p_1', \ldots, p_2$ consisting of a contiguous section of $bd(P)$.
Also, illustrates extending ray $p_1r_1$ of cone $C_{p_1}(p_1r_1, p_1q_1)$ to $p'$, and cone $C_{p_1}(p_1p', p_1q_1)$ leading to determining the next cone $C_{p_1}(p_1r_2, p_1q_1)$.}
\label{fig:rhsimpsimp}
\end{figure}

The boundary of $P''$ consists of $P \cap P'$, the maximal sections of $p_1p_2$ that are bounding $P' \backslash Q$, and $bd(Q) \cap P'$.
Refer to Fig.~\ref{fig:rhsimpsimp}.
Note that some of the vertices of $P''$ are incident to $p_1p_2$.
Due to memory constraints, we do not pre-compute $P''$.
Instead, based on the need, we determine the vertex of interest on-the-fly.
With one traversal of $bd(Q)$, we find an arbitrary vertex $r$ of $Q$ located in the half-plane defined by the line induced by $p_1p_2$ that does not contain $P'$. 
When traversing $bd(Q)$ in counterclockwise direction starting at $r$, let $q_1$ (resp. $q_1'$) be the first (resp. last) intersection of $bd(Q)$ with $P'$, and let $q_2$ be the first vertex of $bd(Q)$ following $q_1$ that is located in $P'$.

Next, we describe an algorithm to compute $q_1, q_1'$, and $q_2$.
Starting at $r$, for every edge $e$ that occurs when $bd(Q)$ is traversed, we find the intersection of line segment $p_1p_2$ with $e$.
The point $q_1$ (which is a vertex of $P''$) is the first point of intersection determined.
Analogously, in this traversal of $bd(Q)$, the point $q_1'$ is the last point of intersection of $e$ with an edge of $Q$.
We note that only a portion of the section of boundary from $q_1$ to $q_1'$, when traversing $bd(Q)$ in counterclockwise direction, participates in $bd(P'')$. 
Significantly, no part of this section of $bd(Q)$ from $q_1$ to $q_1'$ can be in any other pocket of $P$.
The vertex $q_2$ of $bd(Q)$ that follows $q_1$ is found by traversing $bd(Q)$ in counterclockwise direction, starting at $q_1$.
Note that $q_1, q_1'$ and $q_2$ are found by traversing $bd(Q)$ only once. 
This traversal, together with finding intersections of edges of $bd(Q)$ with line segment $p_1p_2$, uses only $O(1)$ workspace.
Upon knowing $q_1$ and $q_2$, we invoke the shortest path algorithm in \cite{journals/jocg/2011AsanoMRW} with cone $C_{p_1}(p_1q_2, p_1q_1)$.

Suppose the algorithm in \cite{journals/jocg/2011AsanoMRW} intends to extend the ray $pr_1$ of cone $C_p(pr_1, pr_2)$, where $r_1$ is a vertex of $P''$.
By traversing $bd(P)$ in counterclockwise direction from $p_1$ to $p_2$, we could determine the point of intersection $p'$ of ray $pr_1$ with $bd(P)$.
Analogously, by traversing $bd(Q)$ in counterclockwise direction from $q_1$ to $q_1'$, we could find the point of intersection $q'$ of ray $pr_1$ with $bd(Q)$.
These traversals to find $p'$ and $q'$ use $O(1)$ workspace.
If both $p'$ and $q'$ exist, we determine the one among these two points closest to point $r_1$ along the ray $pr_1$.
Significantly, the section of $bd(P \cap P')$ traversed in counterclockwise direction from $r_1$ in determining $p'$ is not going to be traversed again in the algorithm.
Next, when the algorithm needs to find a ray's intersection with $bd(P \cap P')$, we find this intersection by traversing $bd(P \cap P')$ starting at $p'$ (instead of starting at $p_1$).
The same is true with the sections of $bd(Q \cap P')$ traversed in finding points of intersections with subsequent rays with $bd(Q)$.
In every iteration, a contiguous section of either $bd(P' \cap P)$ or $bd(P' \cap Q)$ is determined not to be considered in later iterations.
In the worst case, for every edge of $bd(P' \cap P)$ (resp. $bd(P' \cap Q)$) exempted from further consideration, $bd(P' \cap Q)$ (resp. $bd(P' \cap P)$) is traversed at most once, leading to the quadratic time complexity in the worst case.
Further, we note, in processing a pocket of $P$, the portions of $bd(P)$ and $bd(Q)$ located in other pockets of $P$ are not considered.
That is, $P$ is split into pockets, and processing each such pocket takes time quadratic in its size while using $O(1)$ workspace. 

Consider any cone $C_p(pr_1, pr_2)$.
Suppose $r_1$ (resp. $r_2$) is incident on $P \cap P'$.
Then, based on the indices of endpoints of the edge on which $r_1$ (resp. $r_2$) is located, in constant time, we determine the location of $p_2$ with respect to cone $C_p(pr_1, pr_2)$. 
The location of $p_2$ defines the portion of $P''$ that is to be ignored from further consideration.
This defines the simply connected region of $P''$ that is to be processed further for finding the shortest path from $p_1$ to $p_2$.
Finding this region involves updating the current cone, the apex of which will be an intermediate vertex on the shortest path from $p_1$ to $p_2$.
Indeed, every intermediate vertex along the shortest path from $p_1$ to $p_2$ in $P''$ is the apex of some cone considered by the algorithm.
Analogously, if $r_1$ is incident on $bd(P' \cap Q)$, we determine the location of $p_2$ with respect to cone $C_p$ in $O(1)$ time using $O(1)$ workspace.
Hence, for our algorithm, unlike in \cite{journals/jocg/2011AsanoMRW}, there is no need to use the trapezoidal decomposition for locating $p_2$ relative to any of the cones considered.
Finally, we note that every vertex output by our algorithm (apex of a cone) is a vertex of $RH(P | Q)$.

\begin{theorem}\label{proof:simpinsimp}
The algorithm to compute the relative hull of a simple polygon $P$ when $P$ is located inside simple polygon $Q$ correctly computes $RH(P | Q)$, and it takes $O(|P|^2+|Q|^2)$ time while using $O(1)$ workspace.
\end{theorem}
 \begin{proof}
The correctness of our algorithm relies on the following: (i) computing vertices of $CH(P)$, for every two adjacent vertices $p_1, p_2$ of $CH(P)$, determining (ii) whether $bd(Q)$ intersects the relative interior of pocket $P'$ with lid $p_1p_2$, and (iii) when $P' \cap bd(Q) \ne \emptyset$, finding a shortest path between $p_1$ and $p_2$.
For (i), our algorithm uses the Jarvis march.
By traversing $bd(Q)$, our algorithm determined (ii), that is, whether $bd(Q)$ intersects $P'$.
To find the shortest path between $p_1$ and $p_2$ in $P''$, we use the modified algorithm from \cite{journals/jocg/2011AsanoMRW}.
 However, since $P''$ cannot be saved in $O(1)$ workspace, the correctness of (iii) directly relies on correctly updating cones in every iteration, which also involves extending rays incident to concave vertices. 
By traversing $bd(P') \cap bd(P)$ and $P' \cap bd(Q)$, we correctly extend a ray wherever the algorithm in \cite{journals/jocg/2011AsanoMRW} needs.
In addition, in updating the cone, as needed by \cite{journals/jocg/2011AsanoMRW}, both $bd(P') \cap bd(P)$ and $P' \cap bd(Q)$ are considered.
Further, by considering the following cases, the location of $p_2$ with respect to a cone $C_p(pr_1, pr_2)$ is determined correctly: 
(a) $p_2$ is located in the simple polygon formed by $pr_1, pr_2$, and the boundary of $P''$ from $r_1$ to $r_2$ is traversed in counterclockwise direction, 
(b) $p_2$ is located in the simple polygon formed by $pr_1$ and $bd(P'')$ from $r_1$ to $p$ in clockwise direction, or
(c) $p_2$ is located in the simple polygon formed by $pr_2$ and $bd(P'')$ from $r_2$ to $p$ in counterclockwise direction.
Noting the location of $p_2$ with respect to both $bd(P) \cap P'$ and $bd(P'') \backslash bd(P)$, $p_2$ is located correctly by the indices of endpoints of the edge on which $r_1$ (or $r_2$) is incident.
These claims hold good though $bd(Q)$ intersects $P'$.
The algorithm in \cite{journals/jocg/2011AsanoMRW} outputs vertices of the shortest path from $p_1$ to $p_2$, as and when they are determined.
That is, the apex of every cone that has a distinct apex from the previous one is output as a vertex on the shortest path from $p_1$ to $p_2$.

The output-sensitive Jarvis march on the vertices of $P$ takes $O(|P|^2)$ time.
For any pocket $P'$, determining whether $bd(Q)$ intersects $P'$ takes $O(|Q|)$ time.
For any ray $pr_1$ with concave vertex $r_1 \in bd(P)$, extending $pr_1$ involves traversing $bd(P \cap P')$ and $bd(P' \cap Q)$ in clockwise directions, while starting at the point on the boundary where the last traversal terminated on each of these.
As argued above, in the worst case, the total cost of all such traversals together takes quadratic time using constant workspace.
Locating $p_2$ with respect to a given cone takes only constant time by looking at the indices of edges on which a ray of that cone is incident.
The rest of the time complexity is due to the algorithm in \cite{journals/jocg/2011AsanoMRW}, which is upper bounded by $O(|P''|^2)$.
Since the algorithm in \cite{journals/jocg/2011AsanoMRW} outputs intermediate vertices defining shortest paths as soon as they are discovered, that is, without saving them, our algorithm uses only $O(1)$ workspace. 
Summing over the work involved in processing all the pockets of edges of $CH(P)$ together yields the time complexity mentioned. 
\end{proof}

\section{Relative hull of a simple polygon $P$ when another simple polygon located alongside $P$}
\label{sect:simpalongsimp}

Given two simple polygons $P$ and $Q$ with $rint(P) \cap Q = \emptyset$ but $rint(CH(P)) \cap Q \ne \emptyset$, the constant workspace algorithm proposed in this section finds the relative hull of $P$ with respect to $Q$, again denoted by $RH(P|Q)$.  
$RH(P|Q)$ is a weakly simple polygon with a perimeter of minimum length such that $P \subseteq RH(P | Q) \subseteq CH(P)$ and $RH(P|Q) \cap rint(Q) = \emptyset$.
For $S'$ being the union of the set comprising the vertices of $P$ and the set comprising the vertices of $Q$, for convenience, we assume no two points in $S'$ have the same $x$-coordinates or $y$-coordinates, and no three points in $S'$ are collinear.

Since the outline of our algorithm is the same as the algorithm devised for this problem in Toussaint~\cite{journals/dcg/1989Toussaint}, we first provide a high-level description of their algorithm. 
An edge $e$ of $CH(P \cup Q)$ is called a {\it bridge} if one endpoint of $e$ is a vertex of $P$ and the other endpoint of $e$ is a vertex of $Q$. 
Based on the number of bridges, their algorithm determines whether $Q$ belongs to a pocket of $P$.
If $Q$ is not contained in any pocket of $P$, there exist two bridges, say $p_1q_1$ and $p_2q_2$, on $CH(P \cup Q)$.
Let $p_1, p_2$ (resp. $q_1, q_2$) be the endpoints of these two bridges on $bd(P)$ (resp. $bd(Q)$).
By applying Lee and Preparata~\cite{journals/networks/LeePrep84}'s algorithm, their algorithm finds the geodesic shortest path $\tau$ from $p_1$ and $p_2$ in the simple polygonal region bounded by these two bridges, a section of $bd(P)$ from $p_1$ and $p_2$, and a section of $bd(Q)$ from $q_1$ to $q_2$.
Next, they consider the case in which $Q$ belongs to a pocket of $P$.
Let $P'$ be the pocket of $P$ containing $Q$, and let $p_1$ and $p_2$ be the endpoints of the lid of $P'$.
Their algorithm computes the geodesic shortest path $\tau$ between $p_1$ and $p_2$ in the annulus $P'' = P' \backslash rint(Q)$.
In both cases, the boundary cycle formed by concatenating $\tau$ with a section of $bd(CH(P))$ from $p_2$ and $p_1$ contains $P$, and the closed region defined by this cycle is shown to be $RH(P | Q)$.
Their algorithm takes $O(|P|+|Q|)$ time in the worst-case, but it uses $O(|P|+|Q|)$ workspace.
In this section, by modifying this algorithm, we devise a constant workspace algorithm for this problem.

\begin{figure}[h!]
    \centering
    \subfigure{
    \includegraphics[width=0.33\textwidth]{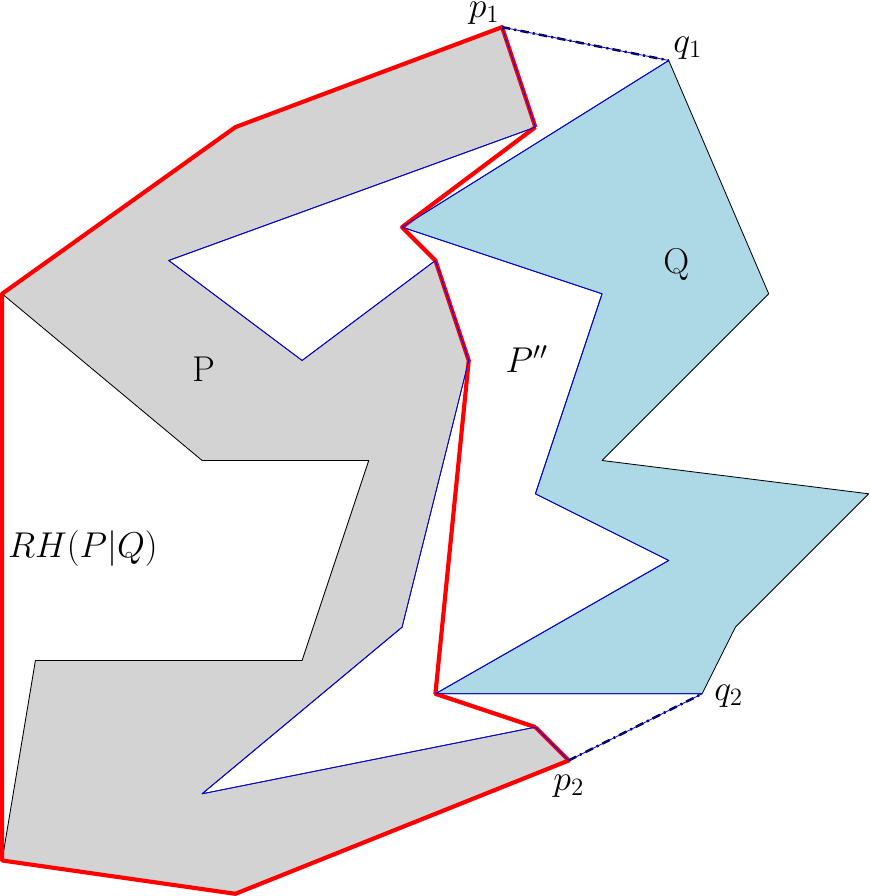}
    }
    \hspace{0.05\textwidth}
    \subfigure{
    \includegraphics[width=0.34\textwidth]{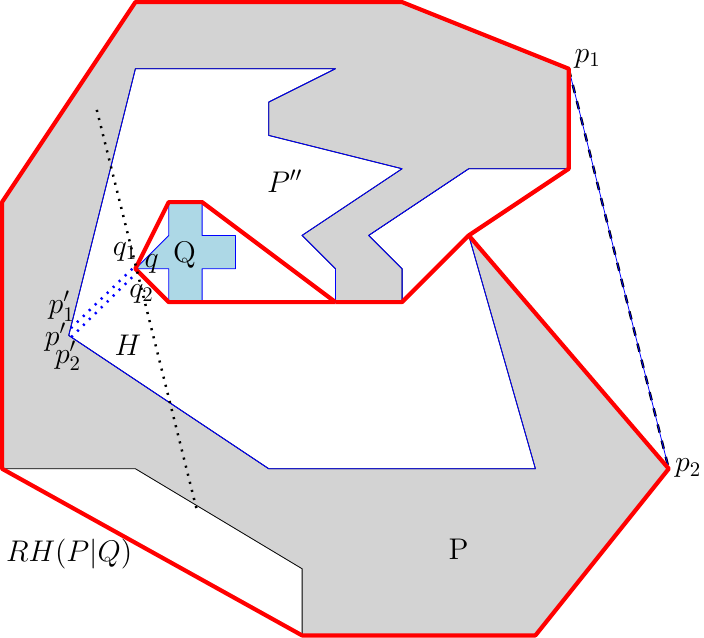}
    }  
    \caption{The figure on the left illustrates case~(i) in which $Q$ does not lie in a pocket of $P$.
    This figure also shows the existence of two bridges $p_1q_1$ and $p_2q_2$.
    The figure on the right illustrates case~(ii) in which $Q$ belongs to a pocket $P'$ of $P$.
    In both the figures, $bd(P'')$ is in blue and $RH(P | Q)$ is in red.
    }
    \label{fig:rhalongside}
\end{figure}

Since $Q$ is placed alongside $P$, from \cite{journals/dcg/1989Toussaint}, the number of bridges on the boundary of $CH(P \cup Q)$ is either zero or two.
Two bridges on $bd(CH(P \cup Q))$ signify case~(i) wherein $Q$ is not contained in $CH(P)$.
In case~(ii), there are no bridges on $bd(CH(P \cup Q))$; this signifies $Q$ belongs to a pocket of $P$, that is, $Q$ is contained in $CH(P)$.
Refer to Fig.~\ref{fig:rhalongside}.
For computing edges of $CH(P \cup Q)$, since Jarvis march uses $O(1)$ workspace, our algorithm applies Jarvis march to the set comprising the vertices of $P$ and the vertices of $Q$.
For every edge computed by this Jarvis march, our algorithm determines whether it is a bridge.
From this, we determine whether the input falls in case~(i) or in case~(ii).
For handling case~(i), our algorithm stores the two bridges it identified, say $p_1q_1$ and $p_2q_2$, wherein $p_1, p_2$ are the vertices of $P$, $q_1, q_2$ are the vertices of $Q$, while $p_1, q_1, q_2, p_2$ occur in that order when $bd(CH(P \cup Q))$ is traversed in clockwise direction starting at $p_1$.
And for handling case~(ii), we will later discuss how our algorithm finds the endpoints $p_1, p_2$ of the lid of the pocket $P'$ that contains $Q$, wherein $p_2$ follows $p_1$ when $bd(CH(P))$ is traversed in clockwise direction starting at $p_1$.

First, we present an algorithm for case~(i).
Let $P''$ be the simple polygon in $CH(P \cup Q) \backslash rint(P \cup Q)$ that has both the bridges $p_1q_1$ and $p_2q_2$ on its boundary.
Let $SP_{P''}(p_1, p_2)$ be the shortest path from $p_1$ to $p_2$ in $P''$.
Also, let $S(p_1, p_2)$ be the section of $bd(CH(P))$ between $p_1$ and $p_2$ such that no point on it belongs to $rint(P'')$. 
The concatenation of $SP_{P''}(p_1, p_2)$ and $S(p_1, p_2)$ is a cycle that contains $P$.
And, the closed region enclosed by this cycle is $RH(P|Q)$. 
Refer to the left subfigure of Fig.~\ref{fig:rhalongside}.
For computing the geodesic shortest path from $p_1$ to $p_2$ in $P''$, we apply the constant workspace algorithm presented in Section~\ref{sect:simpinsimp}.
Though here $p_1$ and $p_2$ are not necessarily visible to each other even in the absence of $Q$, the correctness of that algorithm does not get affected.
Again, instead of computing and saving $P''$, we compute edges of $bd(P'')$ as required by that algorithm. 
Starting at $p_2$, using Jarvis march, we compute edges of $S(p_1, p_2)$, until $p_1$ is encountered.
This takes $O(|P|^2)$ time and uses $O(1)$ workspace.

\begin{lemma}
\label{lem:caseianal}
Given $P$ and $Q$, determining that $Q$ does not belong to $CH(P)$ together with computing $RH(P | Q)$ takes $O(|P|^2+|Q|^2)$ time using $O(1)$ workspace.
\end{lemma}
\begin{proof}
Finding two bridges by computing $CH(P \cup Q)$ with Jarvis march takes $O(|P|^2+|Q|^2)$ time.
For computing $S(p_1, p_2)$, Jarvis march on the vertices of $P$ takes $O(|P|^2)$ time.
The algorithm presented in Section~\ref{sect:simpinsimp} to compute $SP_{P''}(p_1, p_2)$ takes $O(|P|^2+|Q|^2)$ time.
Since Jarvis march uses $O(1)$ workspace, since the algorithm in Section~\ref{sect:simpinsimp} also takes $O(1)$ workspace, and since the vertices of $RH(P|Q)$ are output as and when they are determined, this algorithm is a constant workspace algorithm.
\end{proof}

As mentioned, when there is no bridge on $bd(CH(P \cup Q))$, the input falls into case~(ii).
In this case, our algorithm finds the lid of the pocket in which $Q$ resides. 
For any simple polygon $P'$ and a point $q'$, it is well-known that $q'$ belongs to $rint(P')$ if and only if the number of intersections of the horizontal line passing through $q'$ with $P'$ to the left of $q'$ is odd.
Let $q$ be any vertex of $Q$.
We compute $bd(CH(P))$ using Jarvis march on the vertices of $P$, starting at any vertex of $P$. 
As and when any edge $p_1p_2$ of $bd(CH(P))$ is determined, we check whether that edge is the lid of a pocket that contains $Q$.
This is accomplished by a single traversal of the boundary of the pocket while determining the number of times a horizontal ray from $q$ intersects the boundary of that pocket to the left of $q$.
If this number is odd, then $Q$ belongs to that pocket.
This computation takes time linear in the number of edges of that pocket while using $O(1)$ workspace.
If this pocket does not contain $Q$, then we continue computing $bd(CH(P))$ with Jarvis march, starting at $p_2$.
Since $Q$ has already been found to be in one of the pockets of $P$, one of these iterations is guaranteed to be successful in finding the pocket that contains $Q$.
Since only Jarvis march is used in finding the pocket containing $Q$, only a constant workspace suffices for this computation. 

Let $p_1$ and $p_2$ be the endpoints of the lid of a pocket $P'$ in which $Q$ resides such that $p_2$ follows $p_1$ when traversing $bd(CH(P))$ in clockwise direction.
By traversing $bd(Q)$, we find a vertex $q$ of $Q$ that is farthest from the line passing through $p_1$ and $p_2$.
It is obvious that $q$ is a vertex of $RH(P | Q)$. 
Finding $q$ takes $O(|Q|)$ time using $O(1)$ workspace.
Let $H$ be the half-plane not containing $Q$ defined by the line parallel to $p_1p_2$ that passes through $q$.
By traversing $bd(P)$, we find a vertex $p'$ of $P$ located in $H$ such that the relative interior of $p'q$ belongs to $P' \backslash Q$. 
Since $bd(P')$ is a cycle and since $Q$ is not in $H$, such a $p'$ is guaranteed to exist in $H$. 
Refer to the right subfigure of Fig.~\ref{fig:rhalongside}.
Finding such a $p'$ takes $O(|P|^2)$ time using $O(1)$ workspace.
To facilitate finding shortest paths using the constant workspace algorithm, we define simple polygon $P''$ with the help of line segment $p'q$.
And, since the $bd(P'')$ consists of sections of boundaries of $P$ and $Q$, to find geodesic shortest paths between any two points in $P''$, we use the algorithm presented in Section~\ref{sect:simpinsimp}, which is a modification to the algorithm given in Asano~\etal~\cite{journals/jocg/2011AsanoMRW}.
Let $p_1', p_2'$ be two points that are infinitesimally close to $p'$ on $bd(P)$ such that points $p_1', p', p_2'$ are encountered in that order when $bd(P)$ is traversed in clockwise direction starting at $p_1'$.
Also, let $q_1, q_2$ be two points that are infinitesimally close to $q$ on $bd(Q)$ such that points $q_1, q, q_2$ are encountered in that order when $bd(Q)$ is traversed in counterclockwise direction starting at $q_1$.
Among line segments $p_1'q_1$ and $p_2'q_2$, without loss of generality, suppose $p_1'q_1$ is in $\mathbb{R}^2 \backslash (rint(P) \cup rint(Q))$.
The closed region bounded by line segment $p_1p_2$, polygonal chain from $p_2$ to $p'$ along $bd(P')$ when $bd(P')$ is traversed in clockwise direction starting at $p_2$, line segment $p'q$, polygonal chain from $q$ to $q_1$ along $bd(Q)$ when $bd(Q)$ is traversed in clockwise direction starting at $q$, line segment $q_1p_1'$, and the polygonal chain from $p_1'$ to $p_1$ along $bd(P')$ when $bd(P')$ is traversed in clockwise direction starting at $p_1'$, is a simple polygon $P''$.
Then, we compute two shortest paths in $P''$, each with a constant workspace algorithm presented in Section~\ref{sect:simpinsimp}, one between $p_1$ and $q_1$ and the other between $p_2$ and $q$.
Our algorithm outputs as and when any of the vertices on these shortest paths are computed.
Significantly, for computing these shortest paths, like in the algorithm presented in Section~\ref{sect:simpinsimp}, we do not explicitly compute $P''$.
We note that those vertices shared by the shortest path from $p_1$ to $q_1$ and the shortest path from $p_2$ to $q$ are output two times by our algorithm; however, since the objective is to output the boundary cycle of $RH(P|Q)$, the algorithm is correct in outputting these edges twice.
By applying Jarvis march to the vertices of $P$ that are encountered while traversing $bd(P)$ in clockwise direction starting at $p_2$, we compute the section of $bd(CH(P))$ from $p_2$ to $p_1$, that is, Jarvis march continues until it outputs $p_1$.
Note that both $p_1$ and $p_2$ were already determined to be on $bd(CH(P))$.
Again, as each vertex of this section of $bd(CH(P))$ is determined, we output that vertex.

\begin{lemma}
\label{lem:caseiianal}
Given $P$ and $Q$, determining that $Q$ lies in a pocket of $P$ together with computing $RH(P | Q)$ takes $O(|P|^2+|Q|^2)$ time using $O(1)$ workspace.
\end{lemma}
\begin{proof}
Determining that there are no bridges on $CH(P \cup Q)$ by computing $CH(P \cup Q)$ with Jarvis march takes $O(|P|^2+|Q|^2)$ time using $O(1)$ workspace.
With Jarvis march on the vertices of $P$, computing $CH(P)$ takes $O(|P|^2)$ time using $O(1)$ workspace.
For every edge $e$ computed on $bd(CH(P))$, determining whether $e$ is the lid of a pocket of $P$ takes $O(1)$ time using $O(1)$ workspace.
If $e$ is found to be a lid of the pocket $P'$, determining whether $Q$ lies in the pocket defined by $e$ takes $O(|P'|)$ time using $O(1)$ workspace.
As detailed, this involves locating a vertex $q$ of $Q$ by checking the number of intersections of $bd(P')$ with the horizontal line passing through $q$ to the left of $q$.
Hence, finding the pocket containing $Q$ takes $O(|P|)$ time using $O(1)$ workspace. 
As mentioned, finding vertex $q$ on $bd(Q)$ takes $O(|Q|)$ time using $O(1)$ workspace, and finding $p'$ on $bd(P')$ takes $O(|P'|^2)$ time using $O(1)$ workspace. 
By applying the constant workspace algorithm in Section~\ref{sect:simpinsimp}, computing both the shortest paths, one from $p_1$ and $q_1$ and the other from $p_2$ to $q$, in $P''$ together takes $O(|P'|^2+|Q|^2)$ time.
Since $P''$ is not computed and saved in memory, instead, a vertex or a point on an edge of $P''$ is computed as and when needed by the shortest path finding algorithm, computing these shortest paths uses only $O(1)$ workspace.
Further, since every vertex on these two shortest paths is output as and when it is computed, no additional workspace is required to store any of these vertices.
Computing the section of $bd(CH(P))$ from $p_2$ to $p_1$ with Jarvis march, by traversing $bd(P)$ in clockwise direction starting at $p_2$, takes $O(|P|^2)$ time.
\end{proof}

The following theorem is due to Lemma~\ref{lem:caseianal} and Lemma~\ref{lem:caseiianal}.

\begin{theorem}
Given two simple polygons $P$ and $Q$, with $Q$ located alongside $P$, the algorithm described above correctly computes $RH(P|Q)$ in $O(|P|^2+|Q|^2)$ time using $O(1)$ workspace.
\end{theorem}

\section{Relative hull of a set $S$ of points located in a simple polygon} 
\label{sect:pointsinsimp}

In this section, we seek to find the relative hull $RH(S|P)$ of a set $S$ of points located in the simple polygon $P$.
Here, $RH(S|P)$ is a weakly simple polygon with a perimeter of minimum length that lies in $P$ and contains all the points in $S$.
For $S'$ being the union of the set comprising all the points in $S$ and the set comprising the vertices of $P$, for convenience, we assume no two points in $S'$ have the same $x$-coordinates or $y$-coordinates, and no three points in $S'$ are collinear.

For this problem, Toussaint~\cite{journals/sigproc/1986toussaint} proposed an algorithm, that takes $O((|P| + |S|)\lg{(|P|+|S|)})$ time in the worst-case while using $O(|P|+|S|)$ workspace.
Their algorithm computes a triangulation $T(P)$ of $P$ and uses the Eulerian tour of the dual graph $D(P)$ of $T(P)$ to visit triangles of $T(P)$. 
Whenever a triangle is visited in that tour, the convex hull of points of $S$ located in that triangle is computed.
For every two hulls $h', h''$ that occur next to each other in the nodes corresponding to triangles visited by the Eulerian tour of $D(P)$, a geodesic shortest path is computed between two special points of $S$, one is a vertex of $h'$ and the other one is a vertex of $h''$.
(We precisely define these points in describing our algorithm for this problem.)
All such hulls and the geodesic shortest paths between special vertices of those hulls together define the boundary of a weakly simple polygon $P'$.
Since $P'$ contains $S$ and is a subset of $RH(S|P)$, $RH(P'|P)$ is equal to $RH(S|P)$.
Thus their algorithm outputs the weakly simple polygon $RH(P'|P)$.

While following the outline of the algorithm given in \cite{journals/sigproc/1986toussaint} for this problem, here we devise a constant workspace algorithm using algorithmic ideas from Asano~\etal~\cite{journals/jocg/2011AsanoMRW} and Asano~\etal~\cite{journals/jgaa/2011AsanoMW}.
These include on-the-fly computation of the constrained Delaunay triangulation of an input simple polygon $P$, denoted by $CDT(P)$, an algorithm to visit the triangles of $CDT(P)$ using the Eulerian walk of the dual graph corresponding to $CDT(P)$, and computing the geodesic shortest path between two chosen points of $S$ in every sleeve.

First, we define $CDT(P)$ from Asano~\etal~\cite{journals/jgaa/2011AsanoMW}.
For any edge $pq$ of $CDT(P)$ and a vertex $r$ of $P$ distinct from $p$ and $q$, the triangle $pqr$ belongs to $CDT(P)$ if and only if
(i) $r$ is visible from both $p$ and $q$, and 
(ii) the relative interior of circle defined by $p, q$, and $r$ does not contain any vertex of $P$ that is visible from $r$.
It is well-known that $CDT(P)$ is unique for any simple polygon $P$, and every edge of $P$ is an edge of $CDT(P)$.
The dual of $CDT(P)$ is the graph $D(P)$ obtained by introducing a node for each triangle of $CDT(P)$ and introducing an edge between every two nodes whose corresponding triangles are adjacent along an edge of $CDT(P)$.
It is known that $D(P)$ is a tree since $P$ is a simple polygon.
Hence, we call $D(P)$ the dual-tree of $CDT(P)$.

Like in \cite{journals/jgaa/2011AsanoMW}, our algorithm computes triangles of $CDT(P)$ on-the-fly, that is, a triangle of $CDT(P)$ is computed as and when it is needed.
Since $D(P)$ is not stored, as part of computing the Eulerian tour $\tau$ of $D(P)$, our algorithm computes the triangles of $CDT(P)$ as required for computing $\tau$. 
That is, given an edge $e$ of $CDT(P)$, the triangle(s) to which $e$ is incident in $CDT(P)$ are computed.
The triangles in $CDT(P)$ are also needed to compute the geodesic shortest path between any two points of $S$ in $P$.
For this as well, any triangle in $CDT(P)$ is computed on-the-fly, as required by the geodesic shortest path finding algorithm.
Since $CDT(P)$ is unique for $P$, re-computing triangles of interest in $CDT(P)$ always yields the identical set of triangles.
Also, since the dual-tree $D(P)$ associated with $CDT(P)$ is unique, re-computing edges incident to any node of $D(P)$ yields the same set of edges.

\begin{figure}
\centering
\includegraphics[width=0.35\textwidth]{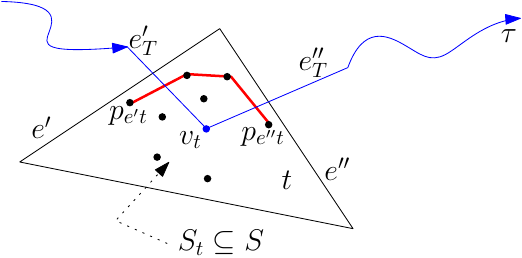}
\caption{The section of boundary of convex hull of the set $S_t \subseteq S$ of points located in triangle $t$ between the connecting point $p_{e't}$ of $e'$ in $t$ and the connecting point $p_{e''t}$ of $e''$ in $t$ is illustrated in red.
The Eulerian tour $\tau$ of dual-tree $D(P)$ of $CDT(P)$ is illustrated in blue. 
Here, $v_t$ (in blue) is the dual node corresponding to triangle $t$. 
Also, $e'$ and $e''$ of $CDT(P)$ respectively correspond to edges $e_T'$ and $e_T''$ of dual-tree $D(P)$.
}
\label{fig:partialhull}
\end{figure}

For every triangle $t \in CDT(P)$ that has points from $S$, for every edge $e$ of $t$, among all the points belonging to $S$ that are located in $t$, following \cite{journals/sigproc/1986toussaint}, the one closest to $e$ is called the {\it connecting point of $e$ in $t$}, denoted by $p_{et}$.
We say a triangle $t$ is visited whenever the Eulerian tour $\tau$ of $D(P)$ visits the dual node $v_t$ of $D(P)$ corresponding to $t$.
Since the degree of any node in $D(P)$ is at most three, in computing $\tau$, any triangle $t$ of $CDT(P)$ is visited at most a constant number of times. 
Suppose the tour $\tau$ encounters a node $v_t \in D(P)$ by traversing an edge $e_T'$ of $D(P)$, and leaves $v_t$ by traversing an edge $e_T''$ of $D(P)$, where $e_T'$ and $e_T''$ are not necessarily distinct edges of $D(P)$ that are incident to $v_t$.
Refer to Fig.~\ref{fig:partialhull}.
Let $S_t \subseteq S$ be the set of points in $t$.
If $|S_t| > 0$ and $e_T' \ne e_T''$, then our algorithm computes the section of $bd(CH(S_t))$ from $p_{e't}$ to $p_{e''t}$ using Jarvis march, where $e'$ is the dual edge of $CDT(P)$ corresponding to $e_T'$ and $e''$ is the dual edge of $CDT(P)$ corresponding to $e_T''$.
As part of computing these hull edges, to determine points in $S$ belonging to $t$, we traverse the input array comprising points in $S$.

\begin{figure}
\centering
\includegraphics[width=0.5\textwidth]{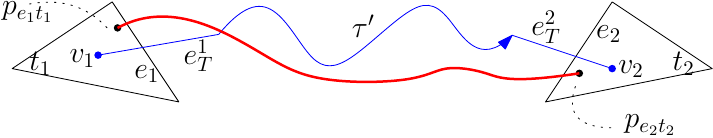}
\caption{Illustrates the geodesic shortest path (in red) from connecting point $p_{e_1t_1}$ to connecting point $p_{e_2t_2}$. 
The subtour $\tau'$ between triangles $t_1$ and $t_2$ in dual-tree $D(P)$ is in blue.
Here, $v_1$ and $v_2$ (in blue) are dual nodes, respectively corresponding to triangles $t_1$ and $t_2$.
Also, triangulation edges $e_1$ and $e_2$ respectively correspond to edges $e_T^1$ and $e_T^2$ in dual-tree $D(P)$.
}
\label{fig:spsleeve}
\end{figure}

Let $t_1$ and $t_2$ be any two triangles such that both $t_1$ and $t_2$ have points from $S$, triangle $t_2$ is visited after visiting triangle $t_1$, and no triangle with points from $S$ is visited between these visits of $t_1$ and $t_2$. 
Let $\tau'$ be the shortest contiguous section of the Eulerian tour $\tau$ corresponding to this traversal of $D(P)$.
Let $v_1$ and $v_2$ be the nodes of $D(P)$ corresponding to $t_1$ and $t_2$, respectively.
In $\tau'$, let $e_T^1$ be the edge incident to $v_1$ and $e_T^2$ be the edge incident to $v_2$.
Let $e_1$ and $e_2$ be the dual edges of $CDT(P)$ corresponding to $e_T^1$ and $e_T^2$, respectively.
Refer to Fig.~\ref{fig:spsleeve}.
Then, we find the geodesic shortest path between connecting points $p_{e_1t_1}$ and $p_{e_2t_2}$ in the sleeve corresponding to all the triangles of $CDT(P)$ that are visited in $\tau'$.

The main overhead in visiting any triangle $t$ of $CDT(P)$ is computing $t$ itself, that is, just before visiting $t$.
Here, we describe our algorithm to compute triangles of $CDT(P)$, as and when each of them is needed.
Let $p', q'$ be the vertices of $P$ such that $p'$ and $q'$ are visible to each other and the line segment $p'q'$ belongs to the untriangulated portion $U$ of $P$. 
Then, the vertex $r'$ belonging to $U$ is determined so that the triangle $p'q'r'$ belongs to $P$ and the relative interior of the circumcircle $C'$ of $p'q'r'$ contains no vertex of $P$ that is visible to $r'$.
Specifically, for every vertex $r$ in $U$, we first check whether both the line segments $p'r$ and $q'r$ belong to $P$.
This is accomplished by checking whether $p'r$ (resp. $q'r$) intersects any edge of $P$ in $O(|P|^2)$ time using $O(1)$ workspace.
If both of these line segments belong to $P$, then for every vertex $s'$ of $P$ in $C'$, we determine whether $s'$ is visible to $r$ by checking whether the line segment $rs'$ intersects any edge of $P$.
If no such $s'$ exists, then the triangle $p'q'r$ belongs to $CDT(P)$.
As mentioned earlier, there always exists a vertex $r'$ corresponding to $p'q'$ such that $p'q'r'$ is a triangle of $CDT(P)$, and hence our algorithm is guaranteed to find $r'$. 
This naive algorithm to find $r'$ takes $O(|P|^2)$ time using $O(1)$ workspace.
Since every edge of $P$ is an edge of some triangle in $CDT(P)$, we invoke the above algorithm with an arbitrary edge $p_1q_1$ of $P$ and find a vertex $r_1$ of $P$ such that the triangle $p_1q_1r_1$ belongs to $CDT(P)$.

Any triangle of $CDT(P)$ required to compute the Eulerian tour $\tau$ of $D(P)$ is determined as and when it is needed in an online fashion.
Since each triangle is visited at most $O(1)$ times, computing all the triangles in the Eulerian tour takes $O(|P|^3)$ time using $O(1)$ workspace.
In addition, to find the geodesic shortest path between two connecting points, we compute the sleeve.
This also requires computing the triangles belonging to that sleeve on-the-fly.
Since all the triangles in each such sleeve are computed only once and since the union of triangles in all such sleeves is a subset of $CDT(P)$, the time to compute all the triangles of $CDT(P)$ that belong to all the sleeves together is $O(|P|^3)$.
Excluding the time complexity to compute triangles of $CDT(P)$, due to the algorithm given in \cite{journals/jgaa/2011AsanoMW}, computing the Eulerian tour takes $O(|P|)$ time using constant workspace.
This leaves us to analyze the time complexity to compute a geodesic shortest path in a sleeve.
Applying the geodesic shortest path finding algorithm from \cite{journals/jocg/2011AsanoMRW}, excluding the time to compute the triangles of $CDT(P)$, the geodesic shortest path computation in all the sleeves of $P$ together takes $O(|P|^2)$ time using $O(1)$ space.

Next, we provide details of computing $RH(P'|P)$, which is essentially $RH(S|P)$.
We find the point $s$ in $S$ that has the largest $y$-coordinate among all the points in $S$.
It is immediate $s$ is a vertex of both $P'$ and $RH(S|P)$.
Since $s$ is located in $P$ and since $bd(P)$ is a cycle, there is a vertex $p$ of $P$ whose $y$-coordinate is larger than the $y$-coordinate of $s$ such that the relative interior of line segment $sp$ intersects neither $bd(P)$ nor $bd(P')$.
Let $s_1$ and $s_2$ be two points infinitesimally close to $s$ on $bd(P')$ such that $s_1, s, s_2$ occur in that order when $bd(P')$ is traversed in clockwise direction.
Also, let $p_1$ and $p_2$ be two points infinitesimally close to $p$ on $bd(P)$ such that $p_1, p, p_2$ occur in that order when $bd(P)$ is traversed in clockwise direction.
Among line segments $p_1s_1$ and $p_2s_2$, without loss of generality, suppose $p_1s_1$ belongs to $P \backslash rint(P')$.
Excluding the region between line segments $s_1p_1$ and $sp$ from $P \backslash P'$ results in a simple polygon, say $P''$.
Noting that $RH(P'|P)$ is the shortest path from $s_1$ to $s$ in $P''$, we compute the geodesic shortest path from $s_1$ to $s$ in $P''$ using the constant workspace algorithm presented in Section~\ref{sect:simpinsimp}, which is a modification to the algorithm given in \cite{journals/jocg/2011AsanoMRW}.
Though the edges of $P$ and the line segments $s_1p_1, sp$ are available beforehand, the algorithm in Section~\ref{sect:simpinsimp} computes the edges of $P'$ on-the-fly as they are needed.
This does not affect the correctness since the edges of $P'$ are computed in our algorithm in the same order as they are required by the geodesic shortest path finding algorithm. 
Every vertex of the geodesic shortest path from $s_1$ to $s$ in $P''$ is sent to the output stream as and when it is computed.
Below, we provide a proof of correctness and the analysis of our algorithm.

\begin{theorem}
Given a set $S$ of points located in a simple polygon $P$, our algorithm computes $RH(S | Q)$ in $O(|P|^3+|S|^2)$ time using $O(1)$ workspace. 
\end{theorem}
\begin{proof}
Since the Eulerian tour $\tau$ of $D(P)$ visits every node of $D(P)$ by computing the corresponding triangles of $CDT(P)$ on-the-fly, every triangle of $CDT(P)$ with points from $S$ is considered.  
Since the maximum degree of any node of $D(P)$ is three, in computing $\tau$, any triangle in $CDT(P)$ is visited $O(1)$ times.
For any triangle $t$ that has a non-empty set $S_t \subseteq S$ of points, if $t$ is visited by entering via $e'$ and exited via edge $e''$, then the section of $bd(CH(S_t))$ between connecting points $p_{e't}$ and $p_{e''t}$ is correctly computed using Jarvis march. 
Considering the number of times $t$ is visited and the edges through which $t$ is entered and exited in each of those times, $bd(CH(S_t))$ is guaranteed to be computed correctly.
Since in the Eulerian tour $\tau$, every pair of connecting points of three edges of triangle $t$ is considered one after the other to compute the section of hull between them, every edge in $bd(CH(S_t))$ is computed exactly once. 
For any two triangles $t', t''$, there is a unique path between the nodes corresponding to them in the dual-tree $D(P)$. 
And, for every two successive non-empty triangles $t_1, t_2$ along $\tau$, algorithm in \cite{journals/jocg/2011AsanoMRW} correctly computes the geodesic shortest path joining the connecting points $p_{e_1t_1}$ and $p_{e_2t_2}$.
These connecting points respectively correspond to the edge through which $\tau$ exited $t_1$ and the edge through which $\tau$ entered $t_2$.
This shortest path computation is again accomplished by computing the triangles of $CDT(P)$ on-the-fly along the sleeve between $t_1$ and $t_2$.
These triangles are computed correctly since we use the naive algorithm to compute triangles of $CDT(P)$.
From these, the correctness of $P'$ is immediate: $P'$ is a weakly simple polygon that contains all the points in $S$, it is contained in $P$, and $P' \subseteq RH(S|P)$.
As the edges of $P'$ are computed, from these edges, the weakly simple polygon $RH(P'|P)$ in $P \backslash P'$ is computed.
The latter is correct since this problem is reduced to finding a geodesic shortest path in simple polygon $P''$.
As a whole, considering the correctness of the algorithm for this problem in \cite{journals/sigproc/1986toussaint}, the computation of $RH(S|P)$ is correct.

Given a triangle $t$ and an edge $e$ of $t$, computing a triangle in $CDT(P)$ adjacent to $t$ that has $e$ on its boundary takes $O(|P|^2)$ time.
Since the triangulation is not stored and computed on-the-fly, given there are $O(|P|)$ triangles in $CDT(P)$, computing the triangles of $CDT(P)$, while computing an Eulerian tour of $D(P)$ takes $O(|P|^3)$ time using $O(1)$ workspace.
Significantly, since the degree of any triangle in $CDT(P)$ is at most three, each triangle in $CDT(P)$ is visited at most a constant number of times in visiting triangles corresponding to nodes of $D(P)$, as part of computing $\tau$.
Hence, the time involved to compute triangles of $CDT(P)$ for computing $\tau$ is as stated. 
Due to the algorithm for computing the Eulerian tour in \cite{journals/jgaa/2011AsanoMW}, the Eulerian tour $\tau$, excluding the time involved in computing triangles of $CDT(P)$, takes $O(|P|)$ time using $O(1)$ workspace.
For any triangle $t_i$ that has the set $S_i \subseteq S$ of points, for every invocation of Jarvis march on $S_i$, we need to find points that belong to $t_i$ by traversing the list of points in $S$.
Hence, Jarvis march to compute any one edge of the hull of points in $S_i$ takes $O(|S|)$ time.
And, excluding the time to compute $\tau$, the time for computing all such edges of hulls together takes $O(|S|^2)$ time.
Using the algorithm to compute geodesic shortest paths from Section~\ref{sect:simpinsimp}, computing the geodesic shortest path between two connecting points takes $O(k^2)$ time if $k$ is the number of triangles belonging to that sleeve.
Since the sum of the number of triangles in all sleeves is $O(|P|)$, the time to compute all geodesic shortest paths in all the sleeves together is $O(|P|^2)$.
As part of computing geodesic shortest paths between connecting points, every sleeve of interest is computed exactly once, and hence each triangle in each sleeve is also computed only once for this computation.
Hence, the computation of triangles of $CDT(P)$ for computing all shortest paths required in the algorithm together takes $O(|P|^3)$ time using $O(1)$ workspace.
With the geodesic shortest path finding algorithm in Section~\ref{sect:simpinsimp}, excluding the time to compute triangles of $CDT(P)$, computing the geodesic shortest path $RH(P'|P)$ from $s_1$ to $s$ in simple polygon $P''$ takes $O(|P|^2+|S|^2)$ time using $O(1)$ workspace.
\end{proof}

\section{Conclusions}
\label{sect:conclu}

This paper devised constant workspace algorithms for three problems:
(i) relative hull of a simple polygon $P$ when $P$ is located in another simple polygon $Q$, 
(ii) relative hull of a simple polygon $P$ when another simple polygon $Q$ is located alongside $P$, and 
(iii) the relative hull of a set $S$ of points when these points are located in a simple polygon $P$.
For the first two problems, we proposed $O(|P|^2+|Q|^2)$ time solutions using constant workspace. 
The algorithm presented for the last problem has $O(|P|^3+|S|^2)$ time complexity using constant workspace.
We feel that, with further research, the cubic dependence on $|P|$ can possibly be reduced.
The algorithms presented here are the first algorithms to provide space-efficient solutions for computing relative hulls.
Providing time-space trade-offs for all these problems would be interesting.
As an extension to the last problem, given a set $S$ of points and a set $L$ of pairwise disjoint line segments, it would be interesting to devise an efficient constant workspace algorithm for computing a weakly simple polygon $R$ with a perimeter of minimum length such that $R$ contains all the points in $S$ and $R \cap rint(\ell) = \emptyset$ for every $\ell \in L$.

\subsection*{Acknowledgement}

This research of R. Inkulu is supported in part by the National Board for Higher Mathematics (NBHM) grant 2011/33/2023NBHM-R\&D-II/16198.

\bibliographystyle{plain}

\end{document}